\documentclass[10pt,aps,prl,twocolumn,superscriptaddress]{revtex4-2}
\usepackage{newtxtext}
\usepackage{graphicx,graphics}
\usepackage{dcolumn}
\usepackage{amsmath,amssymb,amsfonts}
\usepackage{latexsym,verbatim}
\usepackage{bm,dsfont}
\usepackage{xcolor}
\usepackage{multirow}
\usepackage{array}
\usepackage{nicefrac}
\usepackage{float}
\usepackage[braket, qm]{qcircuit}

\usepackage{tikz}
\usetikzlibrary{calc,backgrounds}
\usetikzlibrary{decorations.pathreplacing}
\usetikzlibrary{decorations.pathreplacing, calligraphy}
\tikzset{
  thickbrace/.style={
    decorate, 
    decoration={calligraphic brace, amplitude=5pt}, 
    very thick 
  }
}

\usepackage{physics}

\newcommand{\su}{\mathfrak{su}}

\newcommand{\gens}{\mathcal{G}}
\newcommand{\graphG}{\mathtt{G}}
\newcommand{\graphH}{\mathtt{H}}
\newcommand{\graphJ}{\mathtt{J}}
\newcommand{\graphQ}{\mathtt{Q}}

\newcommand{\ad}{\mathrm{ad}}

\newcommand{\id}{I}

\newcommand{\Aut}{\mathrm{Aut}}

\newtheorem{conjecture}{Conjecture}
\newtheorem{theorem}{Result}
\newtheorem{fact}{Fact}

\usepackage[dvipsnames,svgnames]{xcolor}

\usepackage[breaklinks=true,colorlinks,citecolor=orange,linkcolor=orange,urlcolor=orange]{hyperref}

\begin{document}

\title{
Obstructions to universality in globally controlled qubit graphs
}

\author{Roberto Gargiulo}
\email{r.gargiulo@fz-juelich.de}
\affiliation{Forschungszentrum Jülich GmbH, Peter Grünberg Institute, Quantum Control (PGI-8),
  52425 Jülich,
  Germany}
\affiliation{%
  University of Cologne, Institute for Theoretical Physics (THP),
  50937 Köln,
  Germany
}

\author{Roberto Menta}
\email{roberto.menta@sns.it}
\affiliation{NEST, Scuola Normale Superiore, I-56127 Pisa, Italy}

\author{Vittorio Giovannetti}
\email{vittorio.giovannetti@sns.it}
\affiliation{NEST, Scuola Normale Superiore, I-56127 Pisa, Italy}

\author{Robert Zeier}
\email{r.zeier@fz-juelich.de}
\affiliation{Forschungszentrum Jülich GmbH, Peter Grünberg Institute, Quantum Control (PGI-8),
  52425 Jülich,
  Germany}

\begin{abstract}
Global control offers a promising route to scalable quantum computing. A recent conjecture by Hu~{\it et~al.}~[\href{https://doi.org/10.48550/arXiv.2508.19075}{arXiv:2508.19075}] proposes that any connected qubit graph equipped with global Ising-type interactions and tunable global transverse fields achieves universality if and only if an additional control field breaks every non-trivial automorphism of the
underlying graph. We disprove this conjecture by exhibiting explicit seven- and nine-qubit counterexamples: connected graphs with trivial automorphism group for which the generated Lie algebra is nonetheless not universal. Our analysis reveals that graph automorphisms capture only part of the Hamiltonian symmetry structure: there exist \textit{hidden} symmetries beyond the automorphism group of the graph. Additionally, in the case of non-trivial automorphism group, we find control terms which break the graph symmetries but are still not universal. These findings sharpen the characterization of universality for globally controlled quantum systems.
\end{abstract}

\maketitle

\textit{\bfseries Introduction.}~Globally controlled quantum systems---where all qubits are driven by the same time-dependent field---provide an experimentally appealing but theoretically constrained model of quantum computation~\cite{Lloyd_1993,benjamin_2000, benjamin_2001, benjamin_2003, Fitzsimons2006, Paz-Silva2009, LloydQAOA,morales2020universality}. Their interest lies in avoiding full qubit addressability, replacing it with uniform control fields and fixed interactions, thereby reducing hardware complexity and control overhead~\cite{cesa2023, menta2024globally, cioni2024conveyorbelt, menta2025building, Mitarai2026, riccardi2026}. At the same time, such architectures naturally interpolate between digital and analog paradigms: discrete gate synthesis must emerge from continuous, globally applied dynamics, making controllability both practically relevant and theoretically subtle~\cite{aiudi2026, Bondar2025, cesa2026engineering, menta2026global, White2026, Calliari2026}. A central question is when such restricted controls generate the full Lie algebra $\mathfrak{su}(2^n)$, i.e., achieve universality.

Universality is naturally formulated in terms of the Lie algebra generated by the available Hamiltonians~\cite{Barenco_elementary_1995,Ramakrishna_1995,Lloyd_universality_1996,jurdjevic1972controllability,dalessandro2022}.
Symmetry methods are a central tool for analyzing universality, providing
a necessary and sufficient criterion \cite{Zeier_Schulte-Herbrueggen_2011}.
This symmetry-based perspective has been extended to decide whether prescribed effective
interactions can be simulated by a given set of Hamiltonians~\cite{zeier2015squares,zimboras2015symmetry}.
These methods have been applied to study fermionic systems~\cite{ZZKS14},
variational quantum algorithms~\cite{Kazi_Larocca_Farinati_Coles_Cerezo_Zeier_2025},
and the ground-state reachability in
variational quantum eigensolvers~\cite{singh2025groundstate}.

For Ising-type interactions with global transverse fields, Hu {\it et al.}~\cite{hu2025universal} recently conjectured that universality holds if and only if all non-trivial automorphisms of the interaction graph are broken by suitable control terms.
We show that this criterion is not sufficient for universality. We construct explicit connected seven- and nine-qubit graphs with trivial automorphism group for which the generated dynamical Lie algebra is strictly non universal.
The obstruction arises from symmetry constraints beyond graph automorphisms: an additional symmetry which cannot be represented as a permutation \cite{Kazi_Larocca_Farinati_Coles_Cerezo_Zeier_2025}.

We also consider another example with non-trivial graph automorphisms. There, we can provide controls which break the graph symmetries, 
while still preserving a linear combinations of these.
This demonstrates that breaking graph-theoretic symmetries alone does not characterize universality in globally controlled Ising systems.

We consider $n$ qubits placed on the vertices $V$ of a connected graph
$\graphG=(V,E)$, where the control set is given by simultaneous single-qubit controls combined with an Ising-type interaction on the edges $E$:
\begin{subequations}
\label{eq:generators}
\begin{align}
  \gens_{\graphG} \;=\;
  \Bigl\{\,
    &H_X = \sum_{j\in V} X_j,\quad H_{ZZ} = \sum_{\{i,j\}\in E} Z_i Z_j,
     \label{eq:qaoa} \\ 
    & H_Z = \sum_{j\in V} Z_j
  \,\Bigr\}.
\end{align}
\end{subequations}
The conjecture of Ref.~\cite{hu2025universal} states:

\begin{conjecture}[\cite{hu2025universal}]\label{conj:hu}
Any connected graph $\graphG=(V,E)$ supplemented by global control fields \eqref{eq:generators} realizes universal quantum computation if and only if there exists at least one additional control Hamiltonian $H_{\rm break}$ that breaks all non-trivial automorphisms of the control graph, rendering its automorphism group trivial.
\end{conjecture}

Evidently, in order to guarantee universality, the controls should possess no symmetry, which proves the necessity of the claim.
Then, the ``if'' direction is the non-trivial claim: trivial automorphism group
should {\it guarantee} universality.

Before delving into the details, recall that a symmetry is a matrix $S$
commuting with all Hamiltonian generators given by a set $\gens$ (as in Eq.~\eqref{eq:generators}),
i.e.,
\begin{equation}\label{eq:symmetry}
 \comm{H}{S} = 0 \text{ for all } H \in \gens.
\end{equation}
In particular, all graph automorphisms $P_\sigma$ given by the operator
representation  of qubit permutations $\sigma$ 
are symmetries.
We precisely define the meaning of breaking all non-trivial automorphisms of a graph $\graphG$.
We say that $H_{\rm break}$ \textit{breaks} all symmetries coming from the automorphism group $\Aut(\graphG)$ of $\graphG$ if and only if there is {\it no} qubit
permutation $P_\sigma$ in $\Aut(\graphG)$ which commutes with the additional control $H_{\rm break}$, i.e.,
\begin{equation}\label{eq:graph_symmetry_breaking}
    \comm{P_\sigma}{H_{\rm break}} \neq 0 \text{ for all } \sigma\in\Aut(\graphG) \setminus\{\id\}.
\end{equation}
As an immediate consequence, the only subgroup of $\Aut(\graphG)$ which commutes with $H_{\rm break}$ is the trivial subgroup containing only the identity.
Hence, the appearance of $H_{\rm break}$ renders $\Aut(\graphG)$ \textit{trivial}.
In both examples discussed in Ref.~\cite{hu2025universal}, $\Aut(\graphG)$ consists of the identity and a global reflection symmetry $R$. 
Hence, in that setting $\comm{R}{H_{\rm break}}\neq 0$ is equivalent to breaking all graph symmetries.

We also highlight that a Hamiltonian $H_{\rm break}$ as in Eq.~\eqref{eq:graph_symmetry_breaking} \textit{always} exists for any $\Aut(\graphG)$, however the conjecture states that \textit{any} such symmetry breaking terms suffices to reach universality.
We disprove this direction in the next paragraph, by providing explicit counterexamples.

\textit{\bfseries The counterexamples.}~Reference~\cite{hu2025universal} considers path graphs connecting nearest neighbors with Hamiltonians
as in Eq.~\eqref{eq:generators} and observes that all symmetries in this case are comprised of graph automorphisms, i.e.~the reflection symmetry.
But this does not holds for all graphs and one such class of examples is given by Fig.~\ref{fig:counterexample}
which contains graphs with symmetries but without non-trivial automorphisms.
For the graph $\graphH$ in Fig.~\ref{fig:counterexample}, we confirm with direct
computations in MAGMA~\cite{MR1484478} that it has a trivial automorphism group
but non-trivial Hamiltonian symmetries. In this case,  
all Hamiltonians can be considered as breaking and we first
focus on examples without non-trivial automorphisms:


\begin{theorem}\label{thm:main}
There exists a connected graph $\graphG$ such that (a)
$\graphG$ has no automorphisms besides the identity
and (b) the Lie algebra generated by $\gens_{\graphG}$ is \textit{not} universal.
In particular, Conjecture~\ref{conj:hu} is false.
\end{theorem}

\begin{figure}[t]
    \centering
    \includegraphics[width=0.7\columnwidth]{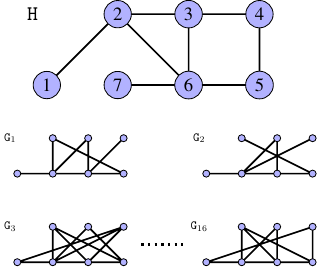}
    \vspace{0.2cm}
    \caption{Examples of the 16 seven-vertex graphs with trivial automorphism group and non-trivial Hamiltonian symmetries over seven vertices. $\graphH$ is the specific example in Result~\ref{thm:main}.}
    \label{fig:counterexample}
\end{figure}

The proof behind Result~\ref{thm:main} relies on finding a specific example, hence we prove it \textit{by construction}.
The graph $\graphH$ shown in Fig.~\ref{fig:counterexample} satisfies the assumptions of Result~\ref{thm:main}. 
In particular, it violates universality by having additional \textit{hidden symmetries}, which do not come from the graph.
Another way of understanding how Result~\ref{thm:main}, in the case of $\graphH$, obstructs universality is in terms of \textit{invariant subspaces}.
Namely, the hidden symmetry structure constrains time evolution inside a given invariant subspace.
In the case of $\graphH$, we find computationally that $\gens_{\graphH}$ respects an invariant subspace decomposition into two invariant subspaces of dimension $2$ and $126$ (see Fig.~\ref{fig:hilbert_space_decomposition}(a)).
Surprisingly, we also found that such a decomposition (in terms of dimensions) also holds for all other $16$ examples on seven vertices from Fig.~\ref{fig:counterexample}.

{\it Proof of Result~\ref{thm:main}.}
Assume first that (a) and (b) are true. Then, since, the automorphism group is already trivial, any $H_{\rm break}$ suffices to break the automorphism group in the sense of Eq.~\eqref{eq:graph_symmetry_breaking}.
As such, to discuss lack of universality it suffices to look at the Lie algebra generated by $\gens_{\graphH}$ itself. 
Here we consider an explicit example which satisfies (a) and such that it has an additional non-trivial Hamiltonian symmetry, hence cannot be universal.
Consider now $\graphH$ as in Fig.~\ref{fig:counterexample}.
To confirm (a) and (b), we use explicit algebraic computations.

(a)~\textit{Trivial automorphism group.}
We enumerate $\Aut(\graphH)$ and confirm $|\Aut(\graphH)|=1$. This can be accomplished using the \textit{nauty} algorithm~\cite{McKay2014}
 as used in MAGMA~\cite{MR1484478}.

(b)~\textit{Non-Universality via Non-trivial Hamiltonian Symmetries.}
To compute the Hamiltonian symmetries for any set of generators $\gens$ it suffices to solve a set of linear homogeneous equations:
\begin{align*}
   & \comm{H_j}{S} = 0 \;\text{ for all }\; H_j\in\gens \iff \\ & (\id {\otimes} H_j - H_j^{t} {\otimes} \id)\; \mathrm{vec}(S) = 0 \;\text{ for all }\;  H_j\in\gens,
\end{align*}
where $\mathrm{vec}(S)$ denotes the vectorization of the symmetry $S$.
Since the matrix elements of the Hamiltonians in $\gens_\graphG$ are integers (hence rational), we can solve this equation exactly.
We find a non-trivial solution, which we can write as
\begin{align*}
        S &= -(\Pi_{13} + \Pi_{27} + \Pi_{46})
        + \Pi_{13}(\Pi_{27} + \Pi_{46})\\
        & \phantom{=\hspace{1mm}} - \Pi_{13}\Pi_{27}\Pi_{46}
\end{align*}
where $\Pi_{ij} = 2\;{\rm SWAP}_{ij}-1$ is (proportional to) the projector onto the singlet state $\ket{00}_{ij}+
\ket{11}_{ij}$, where $\ket{00}_{ij}$ is the product state which is $\ket{0}$ at sites $i$ and $j$.
This concludes the proof.
\hfill$\square$

\begin{figure}[t]
    \centering
    \includegraphics[width=\columnwidth]{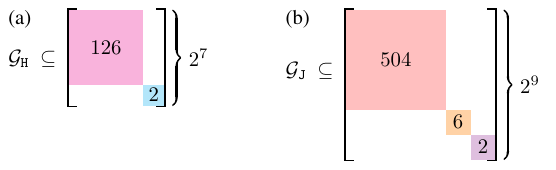}
    \caption{Hilbert space decompositions due to non-trivial hidden symmetries for the graphs (a) $\graphH$ and (b) $\graphJ$.}
    \label{fig:hilbert_space_decomposition}
\end{figure}

\begin{figure}[t]
    \centering
    \includegraphics[width=0.7\columnwidth]{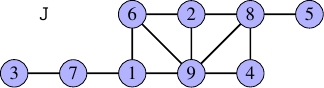}
    \vspace{0.2cm}
    \caption{Nine-vertex graph with trivial automorphism group and non-trivial Hamiltonian symmetries as shown in Fig.~\ref{fig:hilbert_space_decomposition}(b)
    \label{fig:counterexample-two}}
\end{figure}

We mention another graph $\graphJ$ on nine vertices that has only trivial automorphisms (see Fig.~\ref{fig:counterexample-two}).
Considering again the Hamiltonians from Eq.~\eqref{eq:generators},
we observe hidden symmetries spanning a three-dimensional space
and the resulting block
structure is detailed in Fig.~\ref{fig:hilbert_space_decomposition}(b). One feature of this example is that
one might not obtain universality even after adding one non-commuting Hamiltonian as not all of its symmetries need to be broken.

We also highlight that $\graphH$ and $\graphJ$ are only two of the possible graphs which satisfy Result~\ref{thm:main}.
We have also gone through all asymmetric graphs from six to eight vertices.
There are eight asymmetric graphs with six vertices and two of them have hidden symmetries
(but each of these symmetries are contained in the generated Lie algebra).
Out of the 144 asymmetric graphs over seven vertices, we have found 16 graphs for which $\gens_{\graphG}$ has non-trivial Hamiltonian symmetries.
Finally, we have found 228 examples over eight vertices with hidden symmetries, out of a total of 3696 asymmetric graphs.

We now provide an example which evades universality when the automorphism group is non-trivial, independently of any hidden symmetry.

\begin{theorem}\label{thm:main2}
There exists a connected graph $\graphG$ and a Hamiltonian $H_{\mathrm{break}}$ such that: (a) $\Aut(\graphG)$ is non-trivial, (b) there is no non-trivial element of $\Aut(\graphG)$ which commutes with $H_{\mathrm{break}}$ and (c) the Lie algebra generated by $\gens_{\graphG} \cup \{H_{\mathrm{break}}\}$ is \textit{not} universal.
In particular, Conjecture~\ref{conj:hu} is false.
\end{theorem}
{\it Proof of Result~\ref{thm:main2}.}
Here we consider a specific graph $\graphQ$ over $n$ qubits, as given in Fig.~\ref{fig:counterexample2}. It possesses two commuting reflection symmetries $R_1 = {\rm SWAP}_{a,b}$ and $R_2 = {\rm SWAP}_{c,d}$, which exchange the two pairs of qubits at the two ends of the chain ($n>10$).
Hence, (a) is true by construction and the automorphism group is abelian and consists of $\Aut(\graphQ) = \{\id,R_1,R_2,R_1R_2\}$.
\begin{figure}[t]
    \centering
    \includegraphics[width=0.9\columnwidth]{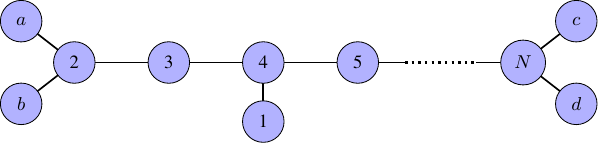}
    \vspace{0.2cm}
    \caption{Example of graph with non-trivial automorphism group. $N>6$ with $n=N+4$ vertices.}
    \label{fig:counterexample2}
\end{figure}
Then, we can consider the following symmetry breaking term:
\begin{align*}
        &H_{\rm break} = H_1 + H_2, \text{ where } \\
        &H_1 = (1+R_2)(X_a - X_b) \text{ and } H_2=(1+R_1)(X_c-X_d).
\end{align*}
It is not hard to check that $H_1$ anti-commutes with $R_1$ and commutes with $R_2$, and viceversa for $H_2$. 
Furthermore, $R_1R_2$ anti-commutes with $H_{\rm break}$.
Hence, $H_{\rm break}$ does not commute with any of the elements in $\Aut(\graphQ)$:
\begin{subequations}
\begin{align*}
        \comm{R_1}{H_{\rm break}} &= 2R_1H_1,\\
        \comm{R_2}{H_{\rm break}} &= 2R_2H_2,\\
        \comm{R_1R_2}{H_{\rm break}} &= 2R_1R_2H_{\rm break},
\end{align*}
\end{subequations}
which proves (b).
However, $H_{\rm break}$ still commutes with a linear combination of them (using $R_2H_1=H_1$, $R_1H_2=H_2$):
\begin{align}
        & \comm{R_1+R_2-R_1R_2}{H_{\rm break}} \nonumber\\
        &= 2R_1H_1 + 2R_2H_2 - 2R_1R_2(H_1+H_2) \nonumber\\
        &= 2R_1H_1 + 2R_2H_2 - 2R_1H_1 - 2R_2H_2\label{eq:symm_example2}
        = 0.
\end{align}
which proves (c). The proof is completed.
\hfill$\square$

\begin{figure}[t]
    \centering
    \includegraphics{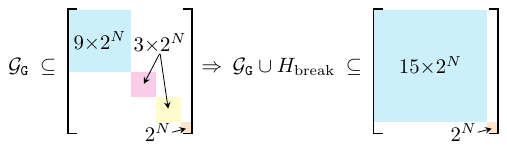}
    \caption{Comparison of Hilbert space decompositions for the graph in Fig.~\ref{fig:counterexample2} based on the automorphisms of the graphs.}
    \label{fig:comparison_decomposition}
\end{figure}

We can also discuss the presence of symmetries from the point of view of invariant subspaces (see Fig.~\ref{fig:comparison_decomposition}). 
The two reflections $R_1$ and $R_2$ alone split the Hilbert space into four subspaces, depending on whether the states are symmetric or anti-symmetric with respect to $R_i$, hence labelled by their eigenvalues, $s_1,s_2\in\{\pm\}$. 
Namely, given that the symmetries are localized at either of the two extreme qubits, we can use the tensor product structure of $R_1$ and $R_2$ to compute the joint eigenspaces.
Each invariant subspace has dimension $d_{s_1}d_{s_2}2^N$, with $d_+ = 3$ the dimension of the locally $R_i$-symmetric space (the triplet states) and $d_-=1$ the dimension of the anti-symmetric space (the singlet state).
By adding $H_{\rm break}$ we can join only three of these invariant subspaces, while still leaving the fully anti-symmetric subspace invariant. 
Namely, consider the projectors $\frac{1}{4}(1+s_1R_1)(1+s_2R_2)$ onto the invariant subspaces. 
One can show that the linear combinations of these which also commute with $H_{\rm break}$ are spanned by the projector onto the fully anti-symmetric subspace $s_1=s_2=-1$ and its complement. 
Explicitly, the projector with $s_1=s_2-1$ is given by $\frac{1}{4}(1-R_1-R_2+R_1R_2)$, which connects to the symmetry found in Eq.~\ref{eq:symm_example2}.
This follows from the (anti-)commutation relations of $H_{\rm break}$ with $R_1,R_2$ and the fact that $R_1H_1$ and $R_2H_2$ are linearly independent.
Hence, up to hidden symmetries (which do not appear in the minimal case $N=7$), the invariant subspaces of $\gens_{\graphQ}\cup H_{\rm break}$ are as in Fig.~\ref{fig:comparison_decomposition}.
This provides an example of symmetry which cannot be broken by simply breaking the permutation symmetries, hence rendering the automorphism group trivial.

\textit{\bfseries Discussion \& open questions.}~We have shown a counterexample to Conjecture~\ref{conj:hu} in the
restricted setting where $\graphG$ has trivial automorphism group.
The obstruction arises from \textit{hidden symmetries}---operators that
commute with the full generator set $\gens_\graphG$ but cannot be
expressed as permutation representations of any graph automorphism.
This directly undermines a generalized version of the construction in Ref.~\cite{hu2025universal}:
the proof of universality therein proceeds by showing that symmetry
breaking forces the dynamical Lie algebra to be all of
$\mathfrak{su}(2^n)$, but this argument implicitly assumes that all
symmetries of $\gens_\graphG$ are accounted for by $\Aut(\graphG)$.
Specifically, they assume that one can always realize the automorphism-symmetric subalgebra of $\su(2^n)$, which in the trivial automorphism case corresponds to the universal Lie algebra.
Our counterexample shows that this is not the case.

Our results establish that breaking all graph automorphisms is \textit{not}
sufficient for universality; a \textit{necessary} condition is that the hidden symmetries be broken as well.
More precisely, any $H_{\rm break}$ that commutes with a non-trivial
Hamiltonian symmetry will leave the
generated Lie algebra strictly below $\mathfrak{su}(2^n)$, regardless
of its action on $\Aut(\graphG)$.
The minimal requirement for universality is therefore that
$H_{\rm break}$ commutes with no non-trivial Hamiltonian symmetry, a condition strictly stronger than
rendering $\Aut(\graphG)$ trivial.
On the other hand, such a statement does not clarify \textit{which} symmetries need to be broken.

However, our numerical examples also highlight that hidden symmetries are not arbitrary. 
Namely, for the 16 examples we found over seven qubits, the Hilbert space breaks down into at most two components, one of which is of small size (always 2). 
This suggests a non-trivial behavior which is reminiscent of \textit{many body scars} \cite{Moudgalya_Bernevig_Regnault_2022}, i.e. special states which are constrained to small dimensional subspaces during time evolution.
Over eight vertices, a richer structure of invariant subspaces already appears, but again with only few invariant subspaces, which are either very large or very small.

For graphs with no hidden symmetries (i.e.~all spanned by $\Aut(\graphG)$), one might hope that breaking all
graph automorphisms would guarantee universality.
However, even 
when $\Aut(\graphG)$ is non-trivial, the Hamiltonian symmetries of the
\textit{extended} generator set $\gens_\graphG \cup \{H_{\rm break}\}$
need not be trivial even if $H_{\rm break}$ satisfies
Eq.~\eqref{eq:graph_symmetry_breaking}: non-trivial symmetries can
survive as linear combinations of the individual graph-symmetry
operators, even when none of those operators individually commutes with
$H_{\rm break}$.
Concretely, if $P_\sigma$ and $P_{\sigma'}$ are two distinct permutation
operators in $\Aut(\graphG)$, it is possible that $\comm{P_\sigma}{H_{\rm
break}} \neq 0$ and $\comm{P_{\sigma'}}{H_{\rm break}} \neq 0$ yet the
linear combination $\alpha P_\sigma + \beta P_{\sigma'}$ still commutes
with the full extended set---see Result~\ref{thm:main2} for details.
This shows that the symmetry-breaking criterion of
Conjecture~\ref{conj:hu} does not exhaust the relevant symmetry
structure even in the absence of hidden symmetries.

Finally, a natural question to ask is how \textit{rare} is the phenomenon of hidden symmetries, and in the absence of these, whether $\gens_{\graphG}$ realizes the automorphism-group symmetric subalgebra.
From our numerical examples, we found that only 16 out of 144 asymmetric graphs on seven
vertices, and 228 out of 3696 asymmetric graphs on eight vertices, exhibit
non-trivial hidden symmetries.
A relevant example, which points to more general behaviors, and has been largely explored, is that of the standard Quantum Approximate Optimization Algorithm (QAOA) setting \cite{Farhi_Goldstone_Gutmann_2014,Kazi_Larocca_Farinati_Coles_Cerezo_Zeier_2025}.There, the generating set consists of the Hamiltonians in Eq.~\eqref{eq:qaoa},
similar to ones in Eq.~\eqref{eq:generators} but without simultaneous single-qubit $Z$-controls. 
As such, they also naturally preserve the \textit{parity} symmetry $X^{\otimes n}$, which sends $Z_i$ into $-Z_i$ and leaves $X_i$ unchanged.
Specifically, in Ref.~\cite{Kazi_Larocca_Farinati_Coles_Cerezo_Zeier_2025}, it was found that with increasing size the asymmetric graphs with hidden symmetries become more and more rare.
It was also proven in Ref.~\cite{Mao_Yuan_Allcock_Zhang_2025} that \textit{almost all graphs} driven with the QAOA controls (Eq.~\eqref{eq:qaoa}) have both no hidden symmetries and coincide with the \textit{parity symmetric subalgebra} of $\su(2^n)$ (up to a phase term).
This also immediately shows that \textit{almost all graphs} in our setting are in fact universal, without the presence of any additional symmetry breaking term.
Indeed, for any graph which in the QAOA setting is universal up to the parity symmetry, any parity breaking term makes it universal.
This is precisely the case of the generating set~\eqref{eq:generators}, since the single qubit $Z$-controls break the parity symmetry.
\begin{fact}[\cite{Mao_Yuan_Allcock_Zhang_2025}]
Almost all connected graphs $\graphG=(V,E)$ supplemented by global control fields \eqref{eq:generators} realize universal quantum computation, without an additional control Hamiltonian $H_{\rm break}$.
\end{fact}
Nevertheless, this does not answer the question of whether practically relevant connectivities also satisfy the conditions for universality discussed in Ref.~\cite{Mao_Yuan_Allcock_Zhang_2025}. We leave it for future investigation.

Finally, we comment on what could be a possible complete set of necessary and sufficient conditions for universality over \textit{all} graphs.
From general properties of semisimple Lie algebras \cite{Zeier_Schulte-Herbrueggen_2011}, it is known that a necessary and sufficient condition for universality of \textit{any generating set} may be understood in a symmetry perspective by looking of the \textit{adjoint representation}, or Liovillian, $\ad_{H_i} = \comm{H_i}{\cdot}$, for $H_i\in\gens$.
Namely, a generating set is universal if and only if the set of \textit{adjoint} symmetries $\comm{\mathcal{S}}{\ad_{H_i}} = 0$, as a vector space, is exactly two-dimensional. 
These symmetries $\mathcal{S}$ are now longer operators on the Hilbert space, but \textit{superoperators}.
Then, the proper symmetry criterion for universality in our setting is the following: $\gens\cup H_{\mathrm{break}}$ is a universal set of controls if and only if $H_{\mathrm{break}}$ breaks all \textit{adjoint} symmetries of $\gens$.
However, much like finding \textit{hidden} symmetries of the Hamiltonians is a non-trivial task in general, we also expect that finding \textit{adjoint} symmetries will not be easy for arbitrary graphs.

\textit{\bfseries Acknowledgments.} 
We acknowledge the use of software MAGMA~\cite{MR1484478}.
This project originated during the ``Quantum at the Dunes'' workshop
held in March 2026 in Natal, Brazil, during a discussion between RG and RM.
We gratefully acknowledge the organizers for providing a stimulating
environment for discussion.
RM thanks H.-Y.~Hu for useful discussions concerning the {\it optimal direct method} developed in Ref.~\cite{hu2025universal}. 
RG and RZ acknowledge funding
under Horizon Europe programme HORIZON-CL4-2022-QUANTUM-02-SGA via the project 
\href{https://doi.org/10.3030/101113690}{101113690} (PASQuanS2.1).
RM and VG acknowledge financial support by MUR (Ministero
dell’Universit\`a e della Ricerca) through the PNRR MUR
project PE0000023-NQSTI.

We note, not without some satisfaction, that the sequence of author
initials R--R--V--R defines a path graph whose automorphism group is
already trivial---so at least the author list is universal.

\bibliography{biblio}

\begin{thebibliography}{37}%
\makeatletter
\providecommand \@ifxundefined [1]{%
 \@ifx{#1\undefined}
}%
\providecommand \@ifnum [1]{%
 \ifnum #1\expandafter \@firstoftwo
 \else \expandafter \@secondoftwo
 \fi
}%
\providecommand \@ifx [1]{%
 \ifx #1\expandafter \@firstoftwo
 \else \expandafter \@secondoftwo
 \fi
}%
\providecommand \natexlab [1]{#1}%
\providecommand \enquote  [1]{``#1''}%
\providecommand \bibnamefont  [1]{#1}%
\providecommand \bibfnamefont [1]{#1}%
\providecommand \citenamefont [1]{#1}%
\providecommand \href@noop [0]{\@secondoftwo}%
\providecommand \href [0]{\begingroup \@sanitize@url \@href}%
\providecommand \@href[1]{\@@startlink{#1}\@@href}%
\providecommand \@@href[1]{\endgroup#1\@@endlink}%
\providecommand \@sanitize@url [0]{\catcode `\\12\catcode `\$12\catcode `\&12\catcode `\#12\catcode `\^12\catcode `\_12\catcode `\%12\relax}%
\providecommand \@@startlink[1]{}%
\providecommand \@@endlink[0]{}%
\providecommand \url  [0]{\begingroup\@sanitize@url \@url }%
\providecommand \@url [1]{\endgroup\@href {#1}{\urlprefix }}%
\providecommand \urlprefix  [0]{URL }%
\providecommand \Eprint [0]{\href }%
\providecommand \doibase [0]{https://doi.org/}%
\providecommand \selectlanguage [0]{\@gobble}%
\providecommand \bibinfo  [0]{\@secondoftwo}%
\providecommand \bibfield  [0]{\@secondoftwo}%
\providecommand \translation [1]{[#1]}%
\providecommand \BibitemOpen [0]{}%
\providecommand \bibitemStop [0]{}%
\providecommand \bibitemNoStop [0]{.\EOS\space}%
\providecommand \EOS [0]{\spacefactor3000\relax}%
\providecommand \BibitemShut  [1]{\csname bibitem#1\endcsname}%
\let\auto@bib@innerbib\@empty
\bibitem [{\citenamefont {Lloyd}(1993)}]{Lloyd_1993}%
  \BibitemOpen
  \bibfield  {author} {\bibinfo {author} {\bibfnamefont {S.}~\bibnamefont {Lloyd}},\ }\bibfield  {title} {\bibinfo {title} {A potentially realizable quantum computer},\ }\href {https://doi.org/10.1126/science.261.5128.1569} {\bibfield  {journal} {\bibinfo  {journal} {Science}\ }\textbf {\bibinfo {volume} {261}},\ \bibinfo {pages} {1569} (\bibinfo {year} {1993})}\BibitemShut {NoStop}%
\bibitem [{\citenamefont {Benjamin}(2000)}]{benjamin_2000}%
  \BibitemOpen
  \bibfield  {author} {\bibinfo {author} {\bibfnamefont {S.~C.}\ \bibnamefont {Benjamin}},\ }\bibfield  {title} {\bibinfo {title} {Schemes for parallel quantum computation without local control of qubits},\ }\href {https://doi.org/10.1103/PhysRevA.61.020301} {\bibfield  {journal} {\bibinfo  {journal} {Phys. Rev. A}\ }\textbf {\bibinfo {volume} {61}},\ \bibinfo {pages} {020301} (\bibinfo {year} {2000})}\BibitemShut {NoStop}%
\bibitem [{\citenamefont {Benjamin}(2001)}]{benjamin_2001}%
  \BibitemOpen
  \bibfield  {author} {\bibinfo {author} {\bibfnamefont {S.~C.}\ \bibnamefont {Benjamin}},\ }\bibfield  {title} {\bibinfo {title} {Quantum computing without local control of qubit-qubit interactions},\ }\href {https://doi.org/10.1103/PhysRevLett.88.017904} {\bibfield  {journal} {\bibinfo  {journal} {Phys. Rev. Lett.}\ }\textbf {\bibinfo {volume} {88}},\ \bibinfo {pages} {017904} (\bibinfo {year} {2001})}\BibitemShut {NoStop}%
\bibitem [{\citenamefont {Benjamin}\ and\ \citenamefont {Bose}(2003)}]{benjamin_2003}%
  \BibitemOpen
  \bibfield  {author} {\bibinfo {author} {\bibfnamefont {S.~C.}\ \bibnamefont {Benjamin}}\ and\ \bibinfo {author} {\bibfnamefont {S.}~\bibnamefont {Bose}},\ }\bibfield  {title} {\bibinfo {title} {Quantum computing with an always-on {Heisenberg} interaction},\ }\href {https://doi.org/10.1103/PhysRevLett.90.247901} {\bibfield  {journal} {\bibinfo  {journal} {Phys. Rev. Lett.}\ }\textbf {\bibinfo {volume} {90}},\ \bibinfo {pages} {247901} (\bibinfo {year} {2003})}\BibitemShut {NoStop}%
\bibitem [{\citenamefont {Fitzsimons}\ and\ \citenamefont {Twamley}(2006)}]{Fitzsimons2006}%
  \BibitemOpen
  \bibfield  {author} {\bibinfo {author} {\bibfnamefont {J.}~\bibnamefont {Fitzsimons}}\ and\ \bibinfo {author} {\bibfnamefont {J.}~\bibnamefont {Twamley}},\ }\bibfield  {title} {\bibinfo {title} {Globally controlled quantum wires for perfect qubit transport, mirroring, and computing},\ }\href {https://doi.org/10.1103/PhysRevLett.97.090502} {\bibfield  {journal} {\bibinfo  {journal} {Phys. Rev. Lett.}\ }\textbf {\bibinfo {volume} {97}},\ \bibinfo {pages} {090502} (\bibinfo {year} {2006})}\BibitemShut {NoStop}%
\bibitem [{\citenamefont {Paz-Silva}\ \emph {et~al.}(2009)\citenamefont {Paz-Silva}, \citenamefont {Brennen},\ and\ \citenamefont {Twamley}}]{Paz-Silva2009}%
  \BibitemOpen
  \bibfield  {author} {\bibinfo {author} {\bibfnamefont {G.~A.}\ \bibnamefont {Paz-Silva}}, \bibinfo {author} {\bibfnamefont {G.~K.}\ \bibnamefont {Brennen}},\ and\ \bibinfo {author} {\bibfnamefont {J.}~\bibnamefont {Twamley}},\ }\bibfield  {title} {\bibinfo {title} {Globally controlled universal quantum computation with arbitrary subsystem dimension},\ }\href {https://doi.org/10.1103/PhysRevA.80.052318} {\bibfield  {journal} {\bibinfo  {journal} {Phys. Rev. A}\ }\textbf {\bibinfo {volume} {80}},\ \bibinfo {pages} {052318} (\bibinfo {year} {2009})}\BibitemShut {NoStop}%
\bibitem [{\citenamefont {Lloyd}(2018)}]{LloydQAOA}%
  \BibitemOpen
  \bibfield  {author} {\bibinfo {author} {\bibfnamefont {S.}~\bibnamefont {Lloyd}},\ }\href@noop {} {\bibinfo {title} {Quantum approximate optimization is computationally universal}} (\bibinfo {year} {2018}),\ \Eprint {https://arxiv.org/abs/1812.11075} {arXiv:1812.11075} \BibitemShut {NoStop}%
\bibitem [{\citenamefont {Morales}\ \emph {et~al.}(2020)\citenamefont {Morales}, \citenamefont {Biamonte},\ and\ \citenamefont {Zimbor{\'a}s}}]{morales2020universality}%
  \BibitemOpen
  \bibfield  {author} {\bibinfo {author} {\bibfnamefont {M.~E.}\ \bibnamefont {Morales}}, \bibinfo {author} {\bibfnamefont {J.}~\bibnamefont {Biamonte}},\ and\ \bibinfo {author} {\bibfnamefont {Z.}~\bibnamefont {Zimbor{\'a}s}},\ }\bibfield  {title} {\bibinfo {title} {On the universality of the quantum approximate optimization algorithm},\ }\href {https://doi.org/10.1007/s11128-020-02748-9} {\bibfield  {journal} {\bibinfo  {journal} {Quantum Inf. Process.}\ }\textbf {\bibinfo {volume} {19}},\ \bibinfo {pages} {1} (\bibinfo {year} {2020})}\BibitemShut {NoStop}%
\bibitem [{\citenamefont {Cesa}\ and\ \citenamefont {Pichler}(2023)}]{cesa2023}%
  \BibitemOpen
  \bibfield  {author} {\bibinfo {author} {\bibfnamefont {F.}~\bibnamefont {Cesa}}\ and\ \bibinfo {author} {\bibfnamefont {H.}~\bibnamefont {Pichler}},\ }\bibfield  {title} {\bibinfo {title} {Universal quantum computation in globally driven {Rydberg} atom arrays},\ }\href {https://doi.org/10.1103/PhysRevLett.131.170601} {\bibfield  {journal} {\bibinfo  {journal} {Phys. Rev. Lett.}\ }\textbf {\bibinfo {volume} {131}},\ \bibinfo {pages} {170601} (\bibinfo {year} {2023})}\BibitemShut {NoStop}%
\bibitem [{\citenamefont {Menta}\ \emph {et~al.}(2025)\citenamefont {Menta}, \citenamefont {Cioni}, \citenamefont {Aiudi}, \citenamefont {Polini},\ and\ \citenamefont {Giovannetti}}]{menta2024globally}%
  \BibitemOpen
  \bibfield  {author} {\bibinfo {author} {\bibfnamefont {R.}~\bibnamefont {Menta}}, \bibinfo {author} {\bibfnamefont {F.}~\bibnamefont {Cioni}}, \bibinfo {author} {\bibfnamefont {R.}~\bibnamefont {Aiudi}}, \bibinfo {author} {\bibfnamefont {M.}~\bibnamefont {Polini}},\ and\ \bibinfo {author} {\bibfnamefont {V.}~\bibnamefont {Giovannetti}},\ }\bibfield  {title} {\bibinfo {title} {Globally driven superconducting quantum computing architecture},\ }\href {https://doi.org/10.1103/PhysRevResearch.7.L012065} {\bibfield  {journal} {\bibinfo  {journal} {Phys. Rev. Res.}\ }\textbf {\bibinfo {volume} {7}},\ \bibinfo {pages} {L012065} (\bibinfo {year} {2025})}\BibitemShut {NoStop}%
\bibitem [{\citenamefont {Cioni}\ \emph {et~al.}(2026)\citenamefont {Cioni}, \citenamefont {Menta}, \citenamefont {Aiudi}, \citenamefont {Polini},\ and\ \citenamefont {Giovannetti}}]{cioni2024conveyorbelt}%
  \BibitemOpen
  \bibfield  {author} {\bibinfo {author} {\bibfnamefont {F.}~\bibnamefont {Cioni}}, \bibinfo {author} {\bibfnamefont {R.}~\bibnamefont {Menta}}, \bibinfo {author} {\bibfnamefont {R.}~\bibnamefont {Aiudi}}, \bibinfo {author} {\bibfnamefont {M.}~\bibnamefont {Polini}},\ and\ \bibinfo {author} {\bibfnamefont {V.}~\bibnamefont {Giovannetti}},\ }\bibfield  {title} {\bibinfo {title} {Conveyor-belt superconducting quantum computer},\ }\href {https://doi.org/10.1103/6zzp-ctyx} {\bibfield  {journal} {\bibinfo  {journal} {Phys. Rev. A}\ }\textbf {\bibinfo {volume} {113}},\ \bibinfo {pages} {012439} (\bibinfo {year} {2026})}\BibitemShut {NoStop}%
\bibitem [{\citenamefont {Menta}\ \emph {et~al.}(2026{\natexlab{a}})\citenamefont {Menta}, \citenamefont {Cioni}, \citenamefont {Aiudi}, \citenamefont {Caravelli}, \citenamefont {Polini},\ and\ \citenamefont {Giovannetti}}]{menta2025building}%
  \BibitemOpen
  \bibfield  {author} {\bibinfo {author} {\bibfnamefont {R.}~\bibnamefont {Menta}}, \bibinfo {author} {\bibfnamefont {F.}~\bibnamefont {Cioni}}, \bibinfo {author} {\bibfnamefont {R.}~\bibnamefont {Aiudi}}, \bibinfo {author} {\bibfnamefont {F.}~\bibnamefont {Caravelli}}, \bibinfo {author} {\bibfnamefont {M.}~\bibnamefont {Polini}},\ and\ \bibinfo {author} {\bibfnamefont {V.}~\bibnamefont {Giovannetti}},\ }\bibfield  {title} {\bibinfo {title} {Building globally controlled quantum processors with {$ZZ$} interactions},\ }\href {https://doi.org/10.1103/lz5d-lnz2} {\bibfield  {journal} {\bibinfo  {journal} {Phys. Rev. A}\ }\textbf {\bibinfo {volume} {113}},\ \bibinfo {pages} {012614} (\bibinfo {year} {2026}{\natexlab{a}})}\BibitemShut {NoStop}%
\bibitem [{\citenamefont {Mitarai}\ \emph {et~al.}(2026)\citenamefont {Mitarai}, \citenamefont {Tadokoro},\ and\ \citenamefont {Tanaka}}]{Mitarai2026}%
  \BibitemOpen
  \bibfield  {author} {\bibinfo {author} {\bibfnamefont {H.}~\bibnamefont {Mitarai}}, \bibinfo {author} {\bibfnamefont {Y.}~\bibnamefont {Tadokoro}},\ and\ \bibinfo {author} {\bibfnamefont {H.}~\bibnamefont {Tanaka}},\ }\bibfield  {title} {\bibinfo {title} {Orthogonal frequency-division multiplexing for simultaneous gate operations on multiple qubits via a shared control line},\ }\href {https://doi.org/10.1103/fmqf-w6ht} {\bibfield  {journal} {\bibinfo  {journal} {Phys. Rev. Appl.}\ }\textbf {\bibinfo {volume} {25}},\ \bibinfo {pages} {044007} (\bibinfo {year} {2026})}\BibitemShut {NoStop}%
\bibitem [{\citenamefont {Riccardi}\ \emph {et~al.}(2026)\citenamefont {Riccardi}, \citenamefont {Glezer~Moshe}, \citenamefont {Menichetti}, \citenamefont {Aiudi}, \citenamefont {Cosenza}, \citenamefont {Abedi}, \citenamefont {Menta}, \citenamefont {Ahmad}, \citenamefont {Nieri~Orfatti}, \citenamefont {Cioni}, \citenamefont {Massarotti}, \citenamefont {Tafuri}, \citenamefont {Giovannetti}, \citenamefont {Polini}, \citenamefont {Caravelli},\ and\ \citenamefont {Szombati}}]{riccardi2026}%
  \BibitemOpen
  \bibfield  {author} {\bibinfo {author} {\bibfnamefont {M.}~\bibnamefont {Riccardi}}, \bibinfo {author} {\bibfnamefont {A.}~\bibnamefont {Glezer~Moshe}}, \bibinfo {author} {\bibfnamefont {G.}~\bibnamefont {Menichetti}}, \bibinfo {author} {\bibfnamefont {R.}~\bibnamefont {Aiudi}}, \bibinfo {author} {\bibfnamefont {C.}~\bibnamefont {Cosenza}}, \bibinfo {author} {\bibfnamefont {A.}~\bibnamefont {Abedi}}, \bibinfo {author} {\bibfnamefont {R.}~\bibnamefont {Menta}}, \bibinfo {author} {\bibfnamefont {H.~G.}\ \bibnamefont {Ahmad}}, \bibinfo {author} {\bibfnamefont {D.}~\bibnamefont {Nieri~Orfatti}}, \bibinfo {author} {\bibfnamefont {F.}~\bibnamefont {Cioni}}, \bibinfo {author} {\bibfnamefont {D.}~\bibnamefont {Massarotti}}, \bibinfo {author} {\bibfnamefont {F.}~\bibnamefont {Tafuri}}, \bibinfo {author} {\bibfnamefont {V.}~\bibnamefont {Giovannetti}}, \bibinfo {author} {\bibfnamefont {M.}~\bibnamefont {Polini}}, \bibinfo {author} {\bibfnamefont {F.}~\bibnamefont {Caravelli}},\ and\ \bibinfo {author} {\bibfnamefont
  {D.}~\bibnamefont {Szombati}},\ }\href@noop {} {\bibinfo {title} {Experimental observation of dynamical blockade between transmon qubits via zz interaction engineering}} (\bibinfo {year} {2026}),\ \Eprint {https://arxiv.org/abs/2601.11714} {arXiv:2601.11714} \BibitemShut {NoStop}%
\bibitem [{\citenamefont {Aiudi}\ \emph {et~al.}(2026)\citenamefont {Aiudi}, \citenamefont {Despres}, \citenamefont {Menta}, \citenamefont {Abedi}, \citenamefont {Menichetti}, \citenamefont {Giovannetti}, \citenamefont {Polini},\ and\ \citenamefont {Caravelli}}]{aiudi2026}%
  \BibitemOpen
  \bibfield  {author} {\bibinfo {author} {\bibfnamefont {R.}~\bibnamefont {Aiudi}}, \bibinfo {author} {\bibfnamefont {J.}~\bibnamefont {Despres}}, \bibinfo {author} {\bibfnamefont {R.}~\bibnamefont {Menta}}, \bibinfo {author} {\bibfnamefont {A.}~\bibnamefont {Abedi}}, \bibinfo {author} {\bibfnamefont {G.}~\bibnamefont {Menichetti}}, \bibinfo {author} {\bibfnamefont {V.}~\bibnamefont {Giovannetti}}, \bibinfo {author} {\bibfnamefont {M.}~\bibnamefont {Polini}},\ and\ \bibinfo {author} {\bibfnamefont {F.}~\bibnamefont {Caravelli}},\ }\bibfield  {title} {\bibinfo {title} {Overcoming disorder in superconducting globally driven quantum computing},\ }\href {https://doi.org/10.1103/zzzc-nqxd} {\bibfield  {journal} {\bibinfo  {journal} {Phys. Rev. A}\ }\textbf {\bibinfo {volume} {113}},\ \bibinfo {pages} {012616} (\bibinfo {year} {2026})}\BibitemShut {NoStop}%
\bibitem [{\citenamefont {Bondar}\ \emph {et~al.}(2025)\citenamefont {Bondar}, \citenamefont {Gaggioli}, \citenamefont {Korpas}, \citenamefont {Marecek}, \citenamefont {Vala},\ and\ \citenamefont {Jacobs}}]{Bondar2025}%
  \BibitemOpen
  \bibfield  {author} {\bibinfo {author} {\bibfnamefont {D.~I.}\ \bibnamefont {Bondar}}, \bibinfo {author} {\bibfnamefont {L.~B.}\ \bibnamefont {Gaggioli}}, \bibinfo {author} {\bibfnamefont {G.}~\bibnamefont {Korpas}}, \bibinfo {author} {\bibfnamefont {J.}~\bibnamefont {Marecek}}, \bibinfo {author} {\bibfnamefont {J.}~\bibnamefont {Vala}},\ and\ \bibinfo {author} {\bibfnamefont {K.}~\bibnamefont {Jacobs}},\ }\bibfield  {title} {\bibinfo {title} {Globally optimal control of quantum dynamics},\ }\href {https://doi.org/10.1103/g4fb-xm13} {\bibfield  {journal} {\bibinfo  {journal} {Phys. Rev. Res.}\ }\textbf {\bibinfo {volume} {7}},\ \bibinfo {pages} {043202} (\bibinfo {year} {2025})}\BibitemShut {NoStop}%
\bibitem [{\citenamefont {Cesa}\ \emph {et~al.}(2026)\citenamefont {Cesa}, \citenamefont {Di~Fini}, \citenamefont {Korbany}, \citenamefont {Tricarico}, \citenamefont {Bernien}, \citenamefont {Pichler},\ and\ \citenamefont {Piroli}}]{cesa2026engineering}%
  \BibitemOpen
  \bibfield  {author} {\bibinfo {author} {\bibfnamefont {F.}~\bibnamefont {Cesa}}, \bibinfo {author} {\bibfnamefont {A.}~\bibnamefont {Di~Fini}}, \bibinfo {author} {\bibfnamefont {D.~A.}\ \bibnamefont {Korbany}}, \bibinfo {author} {\bibfnamefont {R.}~\bibnamefont {Tricarico}}, \bibinfo {author} {\bibfnamefont {H.}~\bibnamefont {Bernien}}, \bibinfo {author} {\bibfnamefont {H.}~\bibnamefont {Pichler}},\ and\ \bibinfo {author} {\bibfnamefont {L.}~\bibnamefont {Piroli}},\ }\href@noop {} {\bibinfo {title} {Engineering discrete local dynamics in globally driven dual-species atom arrays}} (\bibinfo {year} {2026}),\ \Eprint {https://arxiv.org/abs/2601.16961} {arXiv:2601.16961} \BibitemShut {NoStop}%
\bibitem [{\citenamefont {Menta}\ \emph {et~al.}(2026{\natexlab{b}})\citenamefont {Menta}, \citenamefont {Cioni}, \citenamefont {Aiudi}, \citenamefont {Polini},\ and\ \citenamefont {Giovannetti}}]{menta2026global}%
  \BibitemOpen
  \bibfield  {author} {\bibinfo {author} {\bibfnamefont {R.}~\bibnamefont {Menta}}, \bibinfo {author} {\bibfnamefont {F.}~\bibnamefont {Cioni}}, \bibinfo {author} {\bibfnamefont {R.}~\bibnamefont {Aiudi}}, \bibinfo {author} {\bibfnamefont {M.}~\bibnamefont {Polini}},\ and\ \bibinfo {author} {\bibfnamefont {V.}~\bibnamefont {Giovannetti}},\ }\href@noop {} {\bibinfo {title} {Global control via quantum actuators}} (\bibinfo {year} {2026}{\natexlab{b}}),\ \Eprint {https://arxiv.org/abs/2603.23362} {arXiv:2603.23362} \BibitemShut {NoStop}%
\bibitem [{\citenamefont {White}\ \emph {et~al.}(2026)\citenamefont {White}, \citenamefont {Ramesh}, \citenamefont {Impertro}, \citenamefont {Anand}, \citenamefont {Cesa}, \citenamefont {Giudici}, \citenamefont {Iadecola}, \citenamefont {Pichler},\ and\ \citenamefont {Bernien}}]{White2026}%
  \BibitemOpen
  \bibfield  {author} {\bibinfo {author} {\bibfnamefont {R.}~\bibnamefont {White}}, \bibinfo {author} {\bibfnamefont {V.}~\bibnamefont {Ramesh}}, \bibinfo {author} {\bibfnamefont {A.}~\bibnamefont {Impertro}}, \bibinfo {author} {\bibfnamefont {S.}~\bibnamefont {Anand}}, \bibinfo {author} {\bibfnamefont {F.}~\bibnamefont {Cesa}}, \bibinfo {author} {\bibfnamefont {G.}~\bibnamefont {Giudici}}, \bibinfo {author} {\bibfnamefont {T.}~\bibnamefont {Iadecola}}, \bibinfo {author} {\bibfnamefont {H.}~\bibnamefont {Pichler}},\ and\ \bibinfo {author} {\bibfnamefont {H.}~\bibnamefont {Bernien}},\ }\href@noop {} {\bibinfo {title} {Quantum cellular automata on a dual-species {Rydberg} processor}} (\bibinfo {year} {2026}),\ \Eprint {https://arxiv.org/abs/2601.16257} {arXiv:2601.16257} \BibitemShut {NoStop}%
\bibitem [{\citenamefont {Calliari}\ \emph {et~al.}(2026)\citenamefont {Calliari}, \citenamefont {Fromonteil}, \citenamefont {Cesa}, \citenamefont {Zache}, \citenamefont {Preiss}, \citenamefont {Ott},\ and\ \citenamefont {Pichler}}]{Calliari2026}%
  \BibitemOpen
  \bibfield  {author} {\bibinfo {author} {\bibfnamefont {G.}~\bibnamefont {Calliari}}, \bibinfo {author} {\bibfnamefont {C.}~\bibnamefont {Fromonteil}}, \bibinfo {author} {\bibfnamefont {F.}~\bibnamefont {Cesa}}, \bibinfo {author} {\bibfnamefont {T.~V.}\ \bibnamefont {Zache}}, \bibinfo {author} {\bibfnamefont {P.~M.}\ \bibnamefont {Preiss}}, \bibinfo {author} {\bibfnamefont {R.}~\bibnamefont {Ott}},\ and\ \bibinfo {author} {\bibfnamefont {H.}~\bibnamefont {Pichler}},\ }\href@noop {} {\bibinfo {title} {Programmable fermionic quantum processors with globally controlled lattices}} (\bibinfo {year} {2026}),\ \Eprint {https://arxiv.org/abs/2604.13160} {arXiv:2604.13160} \BibitemShut {NoStop}%
\bibitem [{\citenamefont {Barenco}\ \emph {et~al.}(1995)\citenamefont {Barenco}, \citenamefont {Bennett}, \citenamefont {Cleve}, \citenamefont {DiVincenzo}, \citenamefont {Margolus}, \citenamefont {Shor}, \citenamefont {Sleator}, \citenamefont {Smolin},\ and\ \citenamefont {Weinfurter}}]{Barenco_elementary_1995}%
  \BibitemOpen
  \bibfield  {author} {\bibinfo {author} {\bibfnamefont {A.}~\bibnamefont {Barenco}}, \bibinfo {author} {\bibfnamefont {C.~H.}\ \bibnamefont {Bennett}}, \bibinfo {author} {\bibfnamefont {R.}~\bibnamefont {Cleve}}, \bibinfo {author} {\bibfnamefont {D.~P.}\ \bibnamefont {DiVincenzo}}, \bibinfo {author} {\bibfnamefont {N.}~\bibnamefont {Margolus}}, \bibinfo {author} {\bibfnamefont {P.}~\bibnamefont {Shor}}, \bibinfo {author} {\bibfnamefont {T.}~\bibnamefont {Sleator}}, \bibinfo {author} {\bibfnamefont {J.~A.}\ \bibnamefont {Smolin}},\ and\ \bibinfo {author} {\bibfnamefont {H.}~\bibnamefont {Weinfurter}},\ }\bibfield  {title} {\bibinfo {title} {Elementary gates for quantum computation},\ }\href {https://doi.org/10.1103/PhysRevA.52.3457} {\bibfield  {journal} {\bibinfo  {journal} {Phys. Rev. A}\ }\textbf {\bibinfo {volume} {52}},\ \bibinfo {pages} {3457} (\bibinfo {year} {1995})}\BibitemShut {NoStop}%
\bibitem [{\citenamefont {Ramakrishna}\ \emph {et~al.}(1995)\citenamefont {Ramakrishna}, \citenamefont {Salapaka}, \citenamefont {Dahleh}, \citenamefont {Rabitz},\ and\ \citenamefont {Peirce}}]{Ramakrishna_1995}%
  \BibitemOpen
  \bibfield  {author} {\bibinfo {author} {\bibfnamefont {V.}~\bibnamefont {Ramakrishna}}, \bibinfo {author} {\bibfnamefont {M.~V.}\ \bibnamefont {Salapaka}}, \bibinfo {author} {\bibfnamefont {M.}~\bibnamefont {Dahleh}}, \bibinfo {author} {\bibfnamefont {H.}~\bibnamefont {Rabitz}},\ and\ \bibinfo {author} {\bibfnamefont {A.}~\bibnamefont {Peirce}},\ }\bibfield  {title} {\bibinfo {title} {Controllability of molecular systems},\ }\href {https://doi.org/10.1103/PhysRevA.51.960} {\bibfield  {journal} {\bibinfo  {journal} {Phys. Rev. A}\ }\textbf {\bibinfo {volume} {51}},\ \bibinfo {pages} {960} (\bibinfo {year} {1995})}\BibitemShut {NoStop}%
\bibitem [{\citenamefont {Lloyd}(1996)}]{Lloyd_universality_1996}%
  \BibitemOpen
  \bibfield  {author} {\bibinfo {author} {\bibfnamefont {S.}~\bibnamefont {Lloyd}},\ }\bibfield  {title} {\bibinfo {title} {Universal quantum simulators},\ }\href {https://doi.org/10.1126/science.273.5278.1073} {\bibfield  {journal} {\bibinfo  {journal} {Science}\ }\textbf {\bibinfo {volume} {273}},\ \bibinfo {pages} {1073} (\bibinfo {year} {1996})}\BibitemShut {NoStop}%
\bibitem [{\citenamefont {Jurdjevic}\ and\ \citenamefont {Sussmann}(1972)}]{jurdjevic1972controllability}%
  \BibitemOpen
  \bibfield  {author} {\bibinfo {author} {\bibfnamefont {V.}~\bibnamefont {Jurdjevic}}\ and\ \bibinfo {author} {\bibfnamefont {H.~J.}\ \bibnamefont {Sussmann}},\ }\bibfield  {title} {\bibinfo {title} {Control systems on {L}ie groups},\ }\href {https://doi.org/10.1016/0022-0396(72)90035-6} {\bibfield  {journal} {\bibinfo  {journal} {J. Diff. Equat.}\ }\textbf {\bibinfo {volume} {12}},\ \bibinfo {pages} {313} (\bibinfo {year} {1972})}\BibitemShut {NoStop}%
\bibitem [{\citenamefont {D'Alessandro}(2022)}]{dalessandro2022}%
  \BibitemOpen
  \bibfield  {author} {\bibinfo {author} {\bibfnamefont {D.}~\bibnamefont {D'Alessandro}},\ }\href {https://doi.org/10.1201/9781003051268} {\emph {\bibinfo {title} {Introduction to Quantum Control and Dynamics}}},\ \bibinfo {edition} {2nd}\ ed.\ (\bibinfo  {publisher} {CRC Press, Boca Raton},\ \bibinfo {year} {2022})\BibitemShut {NoStop}%
\bibitem [{\citenamefont {Zeier}\ and\ \citenamefont {Schulte-Herbrüggen}(2011)}]{Zeier_Schulte-Herbrueggen_2011}%
  \BibitemOpen
  \bibfield  {author} {\bibinfo {author} {\bibfnamefont {R.}~\bibnamefont {Zeier}}\ and\ \bibinfo {author} {\bibfnamefont {T.}~\bibnamefont {Schulte-Herbrüggen}},\ }\bibfield  {title} {\bibinfo {title} {Symmetry principles in quantum systems theory},\ }\href {https://doi.org/10.1063/1.3657939} {\bibfield  {journal} {\bibinfo  {journal} {J. Math. Phys.}\ }\textbf {\bibinfo {volume} {52}},\ \bibinfo {pages} {113510} (\bibinfo {year} {2011})}\BibitemShut {NoStop}%
\bibitem [{\citenamefont {Zeier}\ and\ \citenamefont {Zimbor{\'a}s}(2015)}]{zeier2015squares}%
  \BibitemOpen
  \bibfield  {author} {\bibinfo {author} {\bibfnamefont {R.}~\bibnamefont {Zeier}}\ and\ \bibinfo {author} {\bibfnamefont {Z.}~\bibnamefont {Zimbor{\'a}s}},\ }\bibfield  {title} {\bibinfo {title} {On squares of representations of compact {Lie} algebras},\ }\href {https://doi.org/10.1063/1.4928410} {\bibfield  {journal} {\bibinfo  {journal} {J. Math. Phys.}\ }\textbf {\bibinfo {volume} {56}},\ \bibinfo {pages} {081702} (\bibinfo {year} {2015})}\BibitemShut {NoStop}%
\bibitem [{\citenamefont {Zimbor{\'a}s}\ \emph {et~al.}(2015)\citenamefont {Zimbor{\'a}s}, \citenamefont {Zeier}, \citenamefont {Schulte-Herbr{\"u}ggen},\ and\ \citenamefont {Burgarth}}]{zimboras2015symmetry}%
  \BibitemOpen
  \bibfield  {author} {\bibinfo {author} {\bibfnamefont {Z.}~\bibnamefont {Zimbor{\'a}s}}, \bibinfo {author} {\bibfnamefont {R.}~\bibnamefont {Zeier}}, \bibinfo {author} {\bibfnamefont {T.}~\bibnamefont {Schulte-Herbr{\"u}ggen}},\ and\ \bibinfo {author} {\bibfnamefont {D.}~\bibnamefont {Burgarth}},\ }\bibfield  {title} {\bibinfo {title} {Symmetry criteria for quantum simulability of effective interactions},\ }\href {https://doi.org/10.1103/PhysRevA.92.042309} {\bibfield  {journal} {\bibinfo  {journal} {Phys. Rev. A}\ }\textbf {\bibinfo {volume} {92}},\ \bibinfo {pages} {042309} (\bibinfo {year} {2015})}\BibitemShut {NoStop}%
\bibitem [{\citenamefont {Zimbor{\'a}s}\ \emph {et~al.}(2014)\citenamefont {Zimbor{\'a}s}, \citenamefont {Zeier}, \citenamefont {Keyl},\ and\ \citenamefont {Schulte-Herbr{\"u}ggen}}]{ZZKS14}%
  \BibitemOpen
  \bibfield  {author} {\bibinfo {author} {\bibfnamefont {Z.}~\bibnamefont {Zimbor{\'a}s}}, \bibinfo {author} {\bibfnamefont {R.}~\bibnamefont {Zeier}}, \bibinfo {author} {\bibfnamefont {M.}~\bibnamefont {Keyl}},\ and\ \bibinfo {author} {\bibfnamefont {T.}~\bibnamefont {Schulte-Herbr{\"u}ggen}},\ }\bibfield  {title} {\bibinfo {title} {A dynamic systems approach to fermions and their relation to spins},\ }\href {https://doi.org/10.1140/epjqt11} {\bibfield  {journal} {\bibinfo  {journal} {EPJ Quantum Technol.}\ }\textbf {\bibinfo {volume} {1}},\ \bibinfo {pages} {11} (\bibinfo {year} {2014})}\BibitemShut {NoStop}%
\bibitem [{\citenamefont {Kazi}\ \emph {et~al.}(2025)\citenamefont {Kazi}, \citenamefont {Larocca}, \citenamefont {Farinati}, \citenamefont {Coles}, \citenamefont {Cerezo},\ and\ \citenamefont {Zeier}}]{Kazi_Larocca_Farinati_Coles_Cerezo_Zeier_2025}%
  \BibitemOpen
  \bibfield  {author} {\bibinfo {author} {\bibfnamefont {S.}~\bibnamefont {Kazi}}, \bibinfo {author} {\bibfnamefont {M.}~\bibnamefont {Larocca}}, \bibinfo {author} {\bibfnamefont {M.}~\bibnamefont {Farinati}}, \bibinfo {author} {\bibfnamefont {P.~J.}\ \bibnamefont {Coles}}, \bibinfo {author} {\bibfnamefont {M.}~\bibnamefont {Cerezo}},\ and\ \bibinfo {author} {\bibfnamefont {R.}~\bibnamefont {Zeier}},\ }\bibfield  {title} {\bibinfo {title} {Analyzing the quantum approximate optimization algorithm: Ans\"atze, symmetries, and {Lie} algebras},\ }\href {https://doi.org/10.1103/yfwq-yqmk} {\bibfield  {journal} {\bibinfo  {journal} {PRX Quantum}\ }\textbf {\bibinfo {volume} {6}},\ \bibinfo {pages} {040345} (\bibinfo {year} {2025})}\BibitemShut {NoStop}%
\bibitem [{\citenamefont {Singh}\ \emph {et~al.}(2025)\citenamefont {Singh}, \citenamefont {Kruckenhauser}, \citenamefont {van Bijnen},\ and\ \citenamefont {Zeier}}]{singh2025groundstate}%
  \BibitemOpen
  \bibfield  {author} {\bibinfo {author} {\bibfnamefont {J.}~\bibnamefont {Singh}}, \bibinfo {author} {\bibfnamefont {A.}~\bibnamefont {Kruckenhauser}}, \bibinfo {author} {\bibfnamefont {R.}~\bibnamefont {van Bijnen}},\ and\ \bibinfo {author} {\bibfnamefont {R.}~\bibnamefont {Zeier}},\ }\href@noop {} {\bibinfo {title} {Ground-state reachability for variational quantum eigensolvers: a {Rydberg}-atom case study}} (\bibinfo {year} {2025}),\ \Eprint {https://arxiv.org/abs/2506.22387} {arXiv:2506.22387} \BibitemShut {NoStop}%
\bibitem [{\citenamefont {Hu}\ \emph {et~al.}(2025)\citenamefont {Hu}, \citenamefont {Gomez}, \citenamefont {Chen}, \citenamefont {Trowbridge}, \citenamefont {Goldschmidt}, \citenamefont {Manchester}, \citenamefont {Chong}, \citenamefont {Jaffe},\ and\ \citenamefont {Yelin}}]{hu2025universal}%
  \BibitemOpen
  \bibfield  {author} {\bibinfo {author} {\bibfnamefont {H.-Y.}\ \bibnamefont {Hu}}, \bibinfo {author} {\bibfnamefont {A.~M.}\ \bibnamefont {Gomez}}, \bibinfo {author} {\bibfnamefont {L.}~\bibnamefont {Chen}}, \bibinfo {author} {\bibfnamefont {A.}~\bibnamefont {Trowbridge}}, \bibinfo {author} {\bibfnamefont {A.~J.}\ \bibnamefont {Goldschmidt}}, \bibinfo {author} {\bibfnamefont {Z.}~\bibnamefont {Manchester}}, \bibinfo {author} {\bibfnamefont {F.~T.}\ \bibnamefont {Chong}}, \bibinfo {author} {\bibfnamefont {A.}~\bibnamefont {Jaffe}},\ and\ \bibinfo {author} {\bibfnamefont {S.~F.}\ \bibnamefont {Yelin}},\ }\href@noop {} {\bibinfo {title} {Universal dynamics with globally controlled analog quantum simulators}} (\bibinfo {year} {2025}),\ \Eprint {https://arxiv.org/abs/2508.19075} {arXiv:2508.19075} \BibitemShut {NoStop}%
\bibitem [{\citenamefont {Bosma}\ \emph {et~al.}(1997)\citenamefont {Bosma}, \citenamefont {Cannon},\ and\ \citenamefont {Playoust}}]{MR1484478}%
  \BibitemOpen
  \bibfield  {author} {\bibinfo {author} {\bibfnamefont {W.}~\bibnamefont {Bosma}}, \bibinfo {author} {\bibfnamefont {J.}~\bibnamefont {Cannon}},\ and\ \bibinfo {author} {\bibfnamefont {C.}~\bibnamefont {Playoust}},\ }\bibfield  {title} {\bibinfo {title} {The {M}agma algebra system. {I}. {T}he user language},\ }\href {https://doi.org/10.1006/jsco.1996.0125} {\bibfield  {journal} {\bibinfo  {journal} {J. Symb. Comput.}\ }\textbf {\bibinfo {volume} {24}},\ \bibinfo {pages} {235} (\bibinfo {year} {1997})},\ \bibinfo {note} {computational algebra and number theory (London, 1993)}\BibitemShut {NoStop}%
\bibitem [{\citenamefont {McKay}\ and\ \citenamefont {Piperno}(2014)}]{McKay2014}%
  \BibitemOpen
  \bibfield  {author} {\bibinfo {author} {\bibfnamefont {B.~D.}\ \bibnamefont {McKay}}\ and\ \bibinfo {author} {\bibfnamefont {A.}~\bibnamefont {Piperno}},\ }\bibfield  {title} {\bibinfo {title} {Practical graph isomorphism, {II}},\ }\href {https://doi.org/10.1016/j.jsc.2013.09.003} {\bibfield  {journal} {\bibinfo  {journal} {J. Symb. Comput.}\ }\textbf {\bibinfo {volume} {60}},\ \bibinfo {pages} {94} (\bibinfo {year} {2014})}\BibitemShut {NoStop}%
\bibitem [{\citenamefont {Moudgalya}\ \emph {et~al.}(2022)\citenamefont {Moudgalya}, \citenamefont {Bernevig},\ and\ \citenamefont {Regnault}}]{Moudgalya_Bernevig_Regnault_2022}%
  \BibitemOpen
  \bibfield  {author} {\bibinfo {author} {\bibfnamefont {S.}~\bibnamefont {Moudgalya}}, \bibinfo {author} {\bibfnamefont {B.~A.}\ \bibnamefont {Bernevig}},\ and\ \bibinfo {author} {\bibfnamefont {N.}~\bibnamefont {Regnault}},\ }\bibfield  {title} {\bibinfo {title} {Quantum many-body scars and {Hilbert} space fragmentation: a review of exact results},\ }\href {https://doi.org/10.1088/1361-6633/ac73a0} {\bibfield  {journal} {\bibinfo  {journal} {Rep. Prog. Phys.}\ }\textbf {\bibinfo {volume} {85}},\ \bibinfo {pages} {086501} (\bibinfo {year} {2022})}\BibitemShut {NoStop}%
\bibitem [{\citenamefont {Farhi}\ \emph {et~al.}(2014)\citenamefont {Farhi}, \citenamefont {Goldstone},\ and\ \citenamefont {Gutmann}}]{Farhi_Goldstone_Gutmann_2014}%
  \BibitemOpen
  \bibfield  {author} {\bibinfo {author} {\bibfnamefont {E.}~\bibnamefont {Farhi}}, \bibinfo {author} {\bibfnamefont {J.}~\bibnamefont {Goldstone}},\ and\ \bibinfo {author} {\bibfnamefont {S.}~\bibnamefont {Gutmann}},\ }\href@noop {} {\bibinfo {title} {A quantum approximate optimization algorithm}} (\bibinfo {year} {2014}),\ \Eprint {https://arxiv.org/abs/1411.4028} {arXiv:1411.4028} \BibitemShut {NoStop}%
\bibitem [{\citenamefont {Mao}\ \emph {et~al.}(2025)\citenamefont {Mao}, \citenamefont {Yuan}, \citenamefont {Allcock},\ and\ \citenamefont {Zhang}}]{Mao_Yuan_Allcock_Zhang_2025}%
  \BibitemOpen
  \bibfield  {author} {\bibinfo {author} {\bibfnamefont {R.}~\bibnamefont {Mao}}, \bibinfo {author} {\bibfnamefont {P.}~\bibnamefont {Yuan}}, \bibinfo {author} {\bibfnamefont {J.}~\bibnamefont {Allcock}},\ and\ \bibinfo {author} {\bibfnamefont {S.}~\bibnamefont {Zhang}},\ }\href@noop {} {\bibinfo {title} {{QAOA-MaxCut} has barren plateaus for almost all graphs}} (\bibinfo {year} {2025}),\ \Eprint {https://arxiv.org/abs/2512.24577} {arXiv:2512.24577} \BibitemShut {NoStop}%
\end{thebibliography}%

\end{document}